\documentclass[11pt]{article}
\usepackage{graphicx}
\usepackage{amsthm}
\usepackage{amsmath}
\usepackage{amssymb}
\usepackage{verbatim}
\usepackage{algorithmic}
\usepackage{algorithm}
\usepackage{color}
\usepackage{enumerate}
\usepackage{appendix}

\usepackage{a4wide}

\newcommand{\eps}{\varepsilon}

\newcommand{\LAS}{\text{\sc{Las}}}

\newcommand{\set}[1]{\left\{ #1 \right\}}
\newcommand{\PS}{\mathcal{P}}

 \newtheorem{theorem}{Theorem}[section]
 
 \newtheorem{lemma}[theorem]{Lemma}

 \newtheorem{definition}{Definition}[section]

\title{A Lasserre Lower Bound for the Min-Sum Single Machine Scheduling Problem
\thanks{Supported by the Swiss National Science Foundation project 200020-144491/1 ``Approximation Algorithms for Machine Scheduling Through Theory and Experiments''.} \footnote{A preliminary version of this paper appeared in 23rd European Symposium on Algorithms - ESA 2015.}}

\author{Adam Kurpisz \and Samuli Lepp\"anen \and Monaldo Mastrolilli\\
{\small \textit{Dalle Molle Institute for Artificial Intelligence Research (IDSIA),}}\\ 
{\small \textit{6928 Manno, Switzerland,}}\\ 
{\small \textit{\{adam,samuli,monaldo\}@idsia.ch}}
\date{}
}

\begin{document}
\maketitle
\begin{abstract}
The Min-sum single machine scheduling problem (denoted $1||\sum f_j$)
generalizes a large number of sequencing problems.
The first constant approximation guarantees have been obtained only recently and are based on natural time-indexed LP relaxations strengthened with
the so called \emph{Knapsack-Cover} inequalities (see Bansal and Pruhs, Cheung and Shmoys and the recent $(4+\epsilon)$-approximation by Mestre and Verschae). These relaxations have an integrality gap of~$2$, since the Min-knapsack problem is a special case. No APX-hardness result is known and it is still conceivable that there exists a PTAS.
Interestingly, the Lasserre hierarchy relaxation, when the objective function is incorporated as a constraint, reduces the integrality gap for the Min-knapsack problem to~$1+\epsilon$.

In this paper we study the complexity of the Min-sum single machine scheduling problem under algorithms from the Lasserre hierarchy. We prove the first lower bound for this model by showing that the integrality gap is unbounded at level $\Omega(\sqrt{n})$ even for a variant of the problem that is solvable in $O(n \log n)$ time by the Moore-Hodgson algorithm, namely Min-number of tardy jobs. We  consider a natural formulation that incorporates the objective function as a constraint and prove the result by partially diagonalizing the matrix associated with the relaxation and exploiting this characterization.
\end{abstract}

\section{Introduction}
The \textsc{Min-sum single machine} scheduling problem (often denoted $1||\sum f_j$) is defined by a set of $n$ jobs to be scheduled on a single machine. Each job has an integral processing time, and there is a monotone function $f_j(C_j)$ specifying the cost incurred when the job $j$ is completed at a particular time $C_j$; the goal is to minimize $\sum f_j(C_j)$.
A natural special case of this problem is given by the \textsc{Min-number of tardy jobs} (denoted $1||\sum w_jU_j$), with $f_j(C_j)= w_j$ if $C_j>d_j$, and $0$ otherwise, where  $w_j\geq 0$, $d_j>0$ are the specific cost and due date of the job $j$ respectively. This problem is known to be NP-complete~\cite{Karp72}. However, restricting to unit weights, the problem can be solved in $O(n \log n)$ time~\cite{Moore68}. %For literature overview we refer to~\cite{MestreVerschae14}

%FPTASs are known with the additional restriction that there are only a constant number of due dates \cite{KarakostasKW12}, or if jobs have unit weights~\cite{lawler82}. There is plenty of literature on special cases of the problem and we refer to~\cite{MestreVerschae14} for an overview.

%A natural special case of this problem is given by the \textsc{Min-sum of tardy jobs} (denoted $1||\sum w_jT_j$), where $f_j(C_j)= w_j \max\{C_j-d_j,0\}$, $w_j\geq 0$, and $d_j>0$ is a specified due date of the job $j$. This problem is known to be strongly NP-complete~\cite{lawler77} and  NP-complete even for unit weights~\cite{Du90} and common due dates~\cite{Yuan92}.FPTASs are known with the additional restriction that there are only a constant number of due dates \cite{KarakostasKW12}, or if jobs have unit weights~\cite{lawler82}. There is plenty of literature on special cases of the problem and we refer to~\cite{MestreVerschae14} for an overview.

%Very recently in~\cite{MegowV13} they consider a special case for the machine changing its speed over time, $1||\sum w_j f$, and obtain a PTAS. For the general $1||\sum_j w_jT_j$ problem, an $n-1$ approximation was presented in~\cite{ChengNY05}.
%

The first constant approximation algorithm for $1||\sum f_j$ was obtained by Bansal and Pruhs \cite{BansalP10}, who considered an even more general scheduling problem. Their $16$-approximation has been recently improved to~${4+\epsilon}$:~Cheung and Shmoys \cite{CheungS11} gave a primal-dual algorithm and claimed that is a $(2+\epsilon)$-approximation; recently, Mestre and Verschae \cite{MestreVerschae14} showed that the analysis in~\cite{CheungS11} cannot yield an approximation better than 4 and provided a proof that the algorithm in \cite{CheungS11} has an approximation ratio of $4 +\epsilon$.

A particular difficulty in approximating this problem lies in the fact that the ratio (\emph{integrality gap})  between the optimal IP solution to the optimal solution of ``natural'' LPs can be arbitrarily large, since the \textsc{Min-knapsack} LP is a common special case.
Thus, in \cite{BansalP10,CheungS11} the authors strengthen natural time-indexed LP relaxations by adding (exponentially many) \emph{Knapsack-Cover} (KC) inequalities introduced  by Wolsey \cite{Wolsey75} (see also \cite{CarrFLP00}) that have proved to be a useful tool to address capacitated covering problems.

%However,no APX-hardness  result is known for $1||\sum w_jT_j$ and, as remarked in \cite{CheungS11}, ``it is still conceivable (and perhaps likely) that there exists a polynomial time approximation scheme'' even for the more general $1||\sum f_j$ problem. This remains one of the most intriguing questions in this area. With this aim, better lower bounds are sought since KC inequalities are not sufficient: indeed, even for the very special case of the \textsc{Min-knapsack} problem, the integrality gap of the natural LP augmented with KC inequalities is $2$~\cite{CarrFLP00}.

%It is known~\cite{CarrFLP00} that even for a very special case of the \textsc{Min knapsack} problem, the LP strengthened with KC inequalities has an integrality gap of 2.
One source of improvements could be the use of semidefinite relaxations such as the powerful Lasserre/Sum-of-Squares hierarchy~\cite{Lasserre01,parrilo00,schor87} (we defer the definition and related results to Section~\ref{sec:Lasserre}). Indeed, it is known~\cite{KarlinMN11} that for \textsc{Min-knapsack} the Lasserre hierarchy relaxation, when the objective function is incorporated as a constraint in the natural LP, reduces the gap to $(1+\eps)$ at level $O(1/\eps)$, for any $\eps>0$.\footnote{The same holds even for the weaker Sherali-Adams hierarchy relaxations.} In light of this observation, it is therefore tempting to understand whether the Lasserre hierarchy relaxation can replace the use of exponentially 
many KC inequalities to get a better approximation for the problem $1||\sum f_j$.\footnote{Note that in order to claim that one can optimize over the Lasserre hierarchy in polynomial time, one needs to assume that  the number of constraint of the starting LP is polynomial in the number of variables (see the discussion in~\cite{Laurent03}). }
%\newpage

In this paper we study the complexity of the \textsc{Min-sum single machine} scheduling problem under algorithms from the Lasserre hierarchy. 
Our contribution is two-fold. We provide a novel technique that is interesting in its own for analyzing integrality gaps for the Lasserre hierarchy. We then use this technique to prove the first lower bound for this model by showing that the integrality gap is unbounded at level $\Omega(\sqrt{n})$ even for the unweighted \textsc{Min-number of tardy jobs} problem, a variant of the problem that admits an $O(n\log n)$ time algorithm~\cite{Moore68}. This result is one of the few known examples where the Lasserre hierarchy requires a non-constant number of levels to exactly solve a problem that admits a polynomial time algorithm. Another well-known such example is the \textsc{Matching} problem, where the Lasserre hierarchy is known to exhibit a vanishing gap at $\Omega(n)$ levels~\cite{Grigoriev01b}. 

This is obtained by formulating the hierarchy as a sum of (exponentially many) rank-one matrices  (Section~\ref{sec:Lasserre}) and, for every constraint, by choosing a dedicated collection (Section~\ref{sec:perturbed}) of rank-one matrices whose sum can be shown to be positive definite by diagonalizing it; it is then sufficient to compare its smallest eigenvalue to the smallest eigenvalue of the remaining part of the sum of the rank-one matrices (Theorem~\ref{lem:integrality_gap_tardy_jobs_t-perturbed}). Furthermore, we complement the result by proving a tight characterization of the considered instance by analyzing the sign of the Rayleigh quotient 
%for the almost diagonal matrix characterization of the Lasserre hierarchy 
(Theorem~\ref{lem:integrality_gap_tardy_jobs_SILP}).
%To the best of our knowledge the integrality gap result for \textsc{Min-number of tardy jobs} is the only known integrality gap result for Lasserre hierarchy with non-vanishing gap for problems with a starting formulations having the objective incorporated as a constraint.

Finally, we show a different use of the above technique to prove that the class of unconstrained $k$ ($\leq n$) degree 0/1 $n$-variate polynomial optimization problems cannot be solved exactly within $k-1$ levels of the  Lasserre hierarchy relaxation. We do this by exhibiting for each $k$ a 0/1 polynomial optimization problem of degree $k$ with an integrality gap. This complements the recent results in~\cite{FawziSaundersonParrilo15,KurpiszLM15}: in~\cite{FawziSaundersonParrilo15} it is shown that the Lasserre relaxation does not have any gap at level $\lceil \frac{n}{2} \rceil$ when optimizing $n$-variate 0/1 polynomials of degree 2; in~\cite{KurpiszLM15} the authors of this paper prove that the only polynomials that can have a gap at level $n-1$ must have degree~$n$.

%\subsection{The Lasserre hierarchy}

%%%%%%%%%%%%%%%%  VERSJA III %%%%%%%%%%%%%%%%%%%%%%%%
\section{The Lasserre Hierarchy} \label{sec:Lasserre}
In this section we provide a formal definition of the Lasserre hierarchy~\cite{Lasserre01} together with a brief overview of the literature. We refer the reader to Appendix~\ref{app:lasserre} for an extended discussion of the form of the hierarchy used here.
\paragraph{Related work.}
The Lasserre/Sum-of-Squares hierarchy~\cite{Lasserre01,parrilo00,schor87} is a systematic procedure to strengthen a relaxation for an optimization problem by constructing a sequence of increasingly tight formulations, obtained by adding additional variables and SDP constraints. The hierarchy is parameterized by its level $t$, such that the formulation gets tighter as $t$ increases, and a solution can be found in time $n^{O(t)}$. This approach captures the convex relaxations used in the best available approximation algorithms for a wide variety of optimization problems. Due to space restrictions, we refer the reader to~\cite{Chla12,Laurent03,ODonnellZ13,Rot13} and the references therein.

%For example, a constant number of rounds captures the approximation algorithms of~\cite{AroraRV09,GoemansW95,Lovasz79} and recently Lee, Raghavendra and Steurer~\cite{LeeRS14} proved that $O(1)$ levels of the Lasserre hierarchy is equivalent in power to any polynomial size SDP extended formulation in approximating constraint satisfaction problems. Other approximation guarantees that arise from the first $O(1)$ levels of the Lasserre (or weaker) hierarchy can be found in~\cite{BarakRS11,BateniCG09,Chlamtac07,ChlamtacS08,DBLP:conf/soda/CyganGM13,VegaK07,GuruswamiS11,MagenM09,RaghavendraT12}.  For a more detailed overview on the use of hierarchies in approximation algorithms, see the surveys~\cite{Chla12,Laurent03,laurent09,Rot13} .

The limitations of the Lasserre hierarchy have also been studied, but not many techniques for proving lower bounds are known. Most of the known lower bounds for the hierarchy originated in the works of Grigoriev~\cite{Grigoriev01,Grigoriev01b} (also independently rediscovered later by Schoenebeck~\cite{Schoenebeck08}).
In \cite{Grigoriev01b} it is shown that random 3XOR or 3SAT instances cannot be solved by even $\Omega(n)$ rounds of Lasserre hierarchy. Lower bounds, such as those of~\cite{BhaskaraCVGZ12,Tulsiani09} rely on \cite{Grigoriev01b,Schoenebeck08} plus gadget reductions. For different techniques to obtain lower bounds see ~\cite{BarakCK15,KurpiszLM15,MekaPW15}.

%%%%%%%%%%%%%%%%%%%%%%%%%%%%%%%%%

%%%%%%%%%%%%%%%%%%%%%%%%%%%%%%%%%
\paragraph{Notation and the formal definition.}
In the context of this paper, it is convenient to define the hierarchy in an equivalent form that follows easily from ``standard'' definitions (see e.g.~\cite{Laurent03}) after a change of variables.\footnote{Notice that the used formulation of the Lasserre hierarchy given in Definition~\ref{def:lasserre_alternative} has exponentially many variables $y^n_I$, due to the change of variables. This is not a problem for our purposes, since we are interested in showing an integrality gap rather than solving an optimization problem. }

For the applications that we have in mind, we restrict our discussion to optimization problems with $0/1$-variables and $m$ linear constraints. 
We denote $K = \set{x \in \mathbb{R}^n~|~g_{\ell}(x)\geq 0, \forall \ell\in [m]}$ to be the feasible set of the linear relaxation. We are interested in approximating the convex hull of the integral points in $K$. We refer to the $\ell$-th linear constraint evaluated at the set $I\subseteq[n]$ ($x_i = 1$ for $i \in I$, and $x_i = 0$ for $i \notin I$) as $g_\ell(x_I)$.
For each integral solution $x_I$, where $I\subseteq N$, in the Lasserre hierarchy defined below there is a variable $y_I^n$ that can be interpreted as the ``relaxed'' indicator variable for the solution $x_I$.

For a set $I \subseteq [n]$ and fixed integer $t$, let $\PS_t(I)$ denote the set of the subsets of $I$ of size at most $t$. For simplicity we write $\PS_t([n]) = \PS_t(n)$. Define $d$-\text{zeta vectors}: $Z_I \in \mathbb{R}^{\PS_d(n)}$ for every $I \subseteq [n]$, such that for each $|J| \leq d$,
$
[Z_I]_J = \left\{ \begin{array}{l} 1,  \text{ if } J \subseteq I \\ 0,  \text{ otherwise} \end{array} \right.
$.
In order to keep the notation simple, we do not emphasize the parameter $d$ as the dimension of the vectors should be clear from the context (we can think of the parameter $d$ as either $t$ or $t+1$).

%\small{
\begin{definition} \label{def:lasserre_alternative}
 The Lasserre hierarchy relaxation at level $t$ for the set $K$, denoted by $\LAS_t(K)$, is given by the set of values $y_I^n \in \mathbb{R}$ for $I \subseteq [n]$ that satisfy
%\footnote{Note that the vectors $Z_I$ in \eqref{eq:lasserre_alternative_2} are understood to be have $|\PS_{t+1}(n)|$ elements and in \eqref{eq:lasserre_alternative_3} $|\PS_{t}(n)|$ elements.}
 \begin{eqnarray}
 \sum_{\substack{I \subseteq [n]}} y_I^n&=& 1, \label{eq:lasserre_alternative_1} \\
 \sum_{\substack{I \subseteq [n]}} y^n_I Z_IZ_I^\top &\succeq& 0, \text{ where } Z_I \in \mathbb{R}^{\PS_{t+1}(n)} \label{eq:lasserre_alternative_2}\\
 \sum_{\substack{I \subseteq [n]}} g_\ell(x_I)y^n_I Z_IZ_I^\top &\succeq& 0, ~ \forall \ell\in [m]  \text{, where }  Z_I \in \mathbb{R}^{\PS_{t}(n)}  \label{eq:lasserre_alternative_3}
\end{eqnarray}
\end{definition}
%}
It is straightforward to see that the Lasserre hierarchy formulation given in Definition~\ref{def:lasserre_alternative} is a relaxation of the integral polytope. Indeed consider any feasible integral solution $x_I \in K$ and set $y_I^n=1$ and the other variables to zero. This solution clearly satisfies Condition~\eqref{eq:lasserre_alternative_1}, Condition~\eqref{eq:lasserre_alternative_2} because the rank one matrix $Z_IZ_I^\top$ is positive semidefinite (PSD), and Condition~\eqref{eq:lasserre_alternative_3} since $x_I\in K$.

%
%%%%%%%%%%%%%%%%%%%%%%%%%%%
\section{Partial Diagonalization} \label{sec:perturbed}
%%%%%%%%%%%%%%%%%%%%%%%%%%%
In this section we describe how to partially diagonalize the matrices associated to Lasserre hierarchy. This will be used in the proofs of Theorem~\ref{lem:integrality_gap_tardy_jobs_t-perturbed} and Theorem~\ref{lem:integrality_gap_tardy_jobs_SILP}.

Below we denote by $w_I^n$ either $y_I^n$ or $y_I^ng_\ell(x_I)$.
The following simple observation describes a congruent transformation ($\cong$) to obtain a partial diagonalization of the matrices used in Definition~\ref{def:lasserre_alternative}. We will use this partial diagonalization in our bound derivation.
\begin{lemma} \label{cor:almost_diagonal_form}
Let $\mathcal{C} \subseteq \PS_n(n)$ be a collection of size $|\PS_d(n)|$ (where $d$ is either $t$ or $t+1$). If $\mathcal{C}$ is such that the matrix $Z$ with columns $Z_I$ for every $I \in \mathcal{C}$ is invertible, then
 $$
 \sum_{I \subseteq [n]} w^n_I Z_IZ_I^\top \cong D + \sum_{I \in \mathcal{P}_n(n) \setminus \mathcal{C}} w^n_I Z^{-1}Z_I (Z^{-1}Z_I)^\top
 $$
 where $D$ is a diagonal matrix with entries $w_I^n$, for $I \in \mathcal{C}$.
\end{lemma}
\begin{proof}
 It is sufficient to note that $\sum_{I \in \mathcal{C}} w^n_I Z_IZ_I^\top = ZDZ^\top$.

\end{proof}
Since congruent transformations are known to preserve the sign of the eigenvalues, the above lemma in principle gives us a technique to check whether or not~\eqref{eq:lasserre_alternative_2} and~\eqref{eq:lasserre_alternative_3} are satisfied: show that the sum of the smallest diagonal element of $D$ and the smallest eigenvalue of the matrix $\sum_{I \in [n] \setminus \mathcal{C}} w^n_I Z^{-1}Z_I (Z^{-1}Z_I)^\top$ is non-negative. In what follows we introduce a method to select the collection $\mathcal{C}$ such that the matrix $Z$ is invertible. 
%This framework as given in the following section will be used in proving Theorem~\ref{lem:integrality_gap_tardy_jobs_t-perturbed} and Theorem~\ref{lem:integrality_gap_polynomials_t-perturbed}.

Let $Z_d$ denote the matrix with columns $[Z_d]_I = Z_I$ indexed by sets $I \subseteq [n]$ of size at most $d$. The matrix $Z_d$ is invertible as it is upper triangular with ones on the diagonal. It is straightforward to check that the inverse $Z_d^{-1}$ is given by
%\small{
%\begin{equation}\label{mobiusmatrix}
$\left[Z_d^{-1}\right ]_{I,J}=
%\left\{
%\begin{array}{ll}
(-1)^{|J\setminus I |}$ %    & \text{if } I \subseteq J,\\
 if $I\subseteq J$ and $0$ otherwise (see e.g.~\cite{Laurent03}).
%0 & \text{otherwise}
%\end{array} \right.
%\end{equation}
%}
%
%%%%%%%%%%%%%%%%%%%%%%%%%%%%%
%\paragraph{The Zeta Matrix of $d$-Perturbed Solutions.}
%%%%%%%%%%%%%%%%%%%%%%%%%%%%%
In Lemma \ref{cor:almost_diagonal_form} we require a collection $\mathcal{C}$ such that the matrix, whose columns are the zeta vectors corresponding to elements in $\mathcal{C}$, is invertible. The above indicates that if we take $\mathcal{C}$ to be the set of subsets of $[n]$ with size less or equal to $d$, then this requirement is satisfied. We can think that the matrix $Z_d$ contains as columns the zeta vectors corresponding to the set $\emptyset$ and all the symmetric differences of the set $\emptyset$ with sets of size at most $d$. The observation allows us to generalize this notion: fix a set $S \subseteq [n]$, and define $\mathcal{C}$ to contain all the sets $S \oplus I$ for $|I| \leq d$ (here $\oplus$ denotes the symmetric difference). More formally, consider the following $|\PS_d(n)|\times |\PS_d(n)|$ matrix $Z_{d(S)}$, whose generic entry $I,J\subseteq  \PS_{d}(n)$ is
\begin{equation}\label{zetamatrix}
\left[Z_{d(S)}\right]_{I,J}=
\left\{
\begin{array}{ll}
1   & \text{if } I \subseteq J\oplus S,\\
0 & \text{otherwise}.
\end{array} \right.
\end{equation}
%}
Note that $Z_{d(\emptyset)}=Z_d$. In order to apply Lemma \ref{cor:almost_diagonal_form}, we show that $Z_{d(S)}$ is invertible.
\begin{lemma} \label{lem:Z_ts_inverse}
 Let $A_{d(S)}$ be a $|\PS_d(n)|\times |\PS_d(n)|$ matrix defined as

\begin{equation}
\left[A_{d(S)}\right]_{I,K}=
\left\{
\begin{array}{ll}
(- 1)^{|K\cap S|}   & \text{if } (I\setminus S)\subseteq K \subseteq I \\
0 & \text{otherwise}.
\end{array} \right.
\end{equation}
Then
$Z_{d(S)}^{-1}= Z_d^{-1} A_{d(S)}$.
\end{lemma}
\begin{proof}
The claim follows by proving that $Z_d = A_{d(S)} Z_{d(S)}$. % (it implies  $Z_t^{-1} (A_{t(S)} Z_{t(S)}) = Z_t^{-1} Z_t  = \mathcal{I}$).
The generic entry $(I,J)$ of $A_{d(S)}Z_{d(S)}$  is
\begin{eqnarray*}
\left[A_{d(S)} Z_{d(S)} \right ]_{I,J} = \sum_{K\in \PS_d(N)} [A_{d(S)}]_{I,K} [Z_{d(S)}]_{K,J} = \sum_{\substack{K\in \PS_d(N)\\ (I\setminus S)\subseteq K \subseteq I\\ K \subseteq J\oplus S}} (- 1)^{|K\cap S|}
\end{eqnarray*}

We first note that unless $I \subseteq J \cup S$, the sum is over an empty set, and thus zero. Indeed, assume there exists an element $a \in I, a \notin J \cup S$. Then, since $I \setminus S \subseteq K$, we require that $a \in K$. On the other hand, $K \subseteq S \oplus J$ implies that $a \in S \oplus J$, which contradicts the assumption on $a$, and hence no such $K$ exists.

Since $K \subseteq (J \oplus S) \cap I$, and $I\setminus S \subseteq J$, we can partition $K$ in the form $K = (I\setminus S) \cup H$, where $H$ is any subset of $I \cap (S \setminus J)$. Indeed, it is easy to see that such a $K$ satisfies the conditions of the sum, and that no other choice is possible. Then, the sum becomes of the form
$$
\left[A_{d(S)} Z_{d(S)} \right ]_{I,J} = \sum_{i = 0}^m (-1)^i\binom{m}{i}
$$
where $m$ is the size of the set $S \cap (I \setminus J)$. Therefore, the sum equals 1 if $m = 0$ and 0 otherwise. It follows that $\left[A_{d(S)} Z_{d(S)} \right ]_{I,J} = 1$ if and only if $I \subseteq J$, and 0 otherwise.

\end{proof}
We also give a closed form of the elements of the matrix $Z^{-1}_{d(S)}$.
\begin{lemma}
\label{lem:Z_ts_inverse_explicite}
 For each $I,J\subseteq  \PS_{d}(N)$ the generic entry $(I,J)$ of $Z_{d(S)}^{-1}$ is

\begin{equation}\label{eq:Z_ts_inverse}
\left[Z_{d(S)}^{-1}\right ]_{I,J}= (-1)^{|J \cap S| + |J \setminus I|}
\left\{
\begin{array}{ll}
(-1)^{d-|I \cup J |} \binom{|S \setminus (I \cup J)|-1}{d- |I \cap J|}, & \text{if } I \setminus S\subseteq J\\
0, & \text{otherwise}.
\end{array} \right.
\end{equation}
\end{lemma}
\begin{proof}
From Lemma~\ref{lem:Z_ts_inverse} we know that $Z_{d(S)}^{-1}= Z_d^{-1} A_{d(S)}$, thus
$$
\left[Z_{d(S)}^{-1}\right ]_{I,J} =  \sum_{K\in \PS_d(N)} [Z_d^{-1}]_{I,K}[A_{d(S)}]_{K,J} = \sum_{\substack{K\in \PS_d(N)\\ I \subseteq K\\ K \setminus S \subseteq J \subseteq K} } (-1)^{|K \setminus I|+|J \cap S|}
$$
First, note that $K \setminus S \subseteq J$ implies that $K \subseteq J \cup S$. This with $I \subseteq K$ implies in particular that the sum has no terms unless $I \subseteq J \cup S$.
Next, we see that $I \cup J \subseteq K$, so we can write $K= I \cup J \cup H$ for some set $H$ disjoint from $I \cup J$. Using the first observation we get that $H \subseteq S \setminus (I \cup J)$. Since $K \in \PS_d(N)$, we thus require that $H \in \PS_{d-|I \cup J|}(S \setminus (I \cup J))$.
The above sum then becomes
$$
\left[Z_{d(S)}^{-1}\right ]_{I,J} = (-1)^{|J \cap S| + |J \setminus I|} \sum_{\substack{ H \in  \mathcal{P}_{d-|I \cup J|}(S \setminus (I \cup J))\\ I \subseteq J \cup S} } (-1)^{|H|}
$$
This simplifies to
$$
\left[Z_{d(S)}^{-1}\right ]_{I,J}= (-1)^{|J \cap S| + |J \setminus I|}
\left\{
\begin{array}{ll}
(-1)^{d-|I \cup J |} \binom{|S \setminus (I \cup J)|-1}{d- |I \cap J|},  & \text{if } I \subseteq J \cup S\\
0, & \text{otherwise}
\end{array} \right.
$$
\end{proof}

\section{A Lower Bound for \textsc{Min-Number of Tardy Jobs}}

We consider the single machine scheduling problem to minimize the number of tardy jobs: we are given a set of $n$ jobs, each with a processing time $p_j>0$, and a due date $d_j>0$. We have to sequence the jobs on a single machine such that no two jobs overlap. For each job $j$ that is not completed by its due date, we pay the cost $w_j$. %If a job $j$ completes at time $C_j$, the tardiness $T_j$ of the job $j$ is $\max\{C_j-d_j,0\}$. The scheduling objective is to minimize the total weighted tardiness, i.e., $\sum_j T_j$. This problem admits an FPTAS~\cite{lawler82}.

\subsection{The Starting Linear Program}
Our result is based on the following ``natural'' linear programming (LP) relaxation that is a special case of the LPs used in \cite{BansalP10,CheungS11} (therefore our gap result also holds if we apply those LP formulations). For each job we introduce a variable $x_j\in[0,1]$ with the intended (integral) meaning that $x_j=1$ if and only if the job $j$ completes after its deadline. Then, for any time $s\in\{d_1,\ldots,d_n\}$, the sum of the processing times of the jobs with deadlines less than $s$, and that complete before $s$, must satisfy $\sum_{j:d_j\leq s} (1-x_j)p_j \leq s$. The latter constraint can be rewritten as a capacitated covering constraint, $\sum_{j:d_j\leq t} x_jp_j \geq D_t$, where $D_s:=\sum_{j:d_j\leq s} p_j -s$ represents the \emph{demand} at time $s$. The goal is to minimize $\sum_j w_j x_j$.

\subsection{The Integrality Gap Instance} Consider the following instance with $n=m^2$ jobs of unit costs. The jobs are partitioned into $m$ blocks $N_1, N_2,\ldots, N_m$, each with $m$ jobs. For $i\in[m]$, the jobs belonging to block $N_i$ have the same processing time $P^i$, for $P>1$, and the same deadline $d_i=m\sum_{j=1}^i P^{j}-\sum_{j=1}^i P^{j-1}$.
Then the demand at time $d_i$ is $D_i=\sum_{j=1}^i P^{j-1}$. For any $t\geq 0$, let $T$ be the smallest value that makes $\LAS_t\left(LP(T)\right)$ feasible, where $LP(T)$ is defined as follows for $x_{ij} \in [0,1]$,  for $ i,j\in[m]$:

\begin{subequations}
\label{LP:tardy}
\begin{align}
  LP(T) \hspace{1cm}& \sum_{i =1}^m \sum_{j =1}^m x_{ij}\leq T,\label{eq:LP_tardy_cardconstr}\\
  &\sum_{i =1}^\ell \sum_{j =1}^m x_{ij}\cdot P^i \geq D_\ell, &  \text{for }\ \ell\in[m]\label{eq:LP_tardy_demand}
  \end{align}
\end{subequations}
Note that, for any feasible \emph{integral} solution for $LP(T)$, the smallest $T$ (i.e. the optimal integral value) can be obtained by selecting one job for each block, so the smallest $T$ for integral solutions is $m=\sqrt{n}$. The \emph{integrality gap} of $\LAS_t\left(LP(T)\right)$ (or $LP(T)$)  is defined as the ratio between $\sqrt{n}$ (i.e. the optimal integral value) and the smallest $T$ that makes $\LAS_t\left(LP(T)\right)$ (or $LP(T)$) feasible.
%
%
%\paragraph{$LP(T)$ Integrality Gap.}
It is easy to check that $LP(T)$ has an integrality gap $P$ for any $P\geq 1$: for $T=\sqrt{n}/P$, a feasible fractional solution for $LP(T)$ exists by setting $x_{ij}=\frac{1}{\sqrt{n}P}$. %This solution is feasible and has value equal to $n/P$.

\subsection{Proof of Integrality Gap for $\LAS_t(LP(T))$}
\begin{theorem}
\label{lem:integrality_gap_tardy_jobs_t-perturbed}
For any $k \geq 1$ and $n$ such that $t=\frac{\sqrt{n}}{2k}-\frac{1}{2} \in \mathbb{N}$, the following solution is feasible for $\LAS_t(LP(\sqrt{n}/k))$
\begin{eqnarray}\label{eq:sol_for_lem_t-perturbed}
y_I^n=
\left\{ \begin{array}{ll}
\alpha, & \forall I\in \PS_{2t+1}(n)\\
0, & \text{otherwise}
\end{array}
\right.
\end{eqnarray}
where $\alpha > 0$ is such that $\sum_{I \subseteq [n]} y_I^n = 1$ and the parameter $P$ is large enough.
\end{theorem}
\begin{proof}
We need to show that the solution~\eqref{eq:sol_for_lem_t-perturbed} satisifies the feasibility conditions~\eqref{eq:lasserre_alternative_1}--\eqref{eq:lasserre_alternative_2} for the variables and the condition~\eqref{eq:lasserre_alternative_3} for every constraint. The condition~\eqref{eq:lasserre_alternative_1} is satisfied by definition of the solution, and~\eqref{eq:lasserre_alternative_2} becomes a sum of positive semidefinite matrices $Z_IZ_I^\top$ with non-negative weights $y^n_I$, so it is satisfied as well.

It remains to show that~\eqref{eq:lasserre_alternative_3} is satisfied for both~\eqref{eq:LP_tardy_cardconstr} and~\eqref{eq:LP_tardy_demand}. Consider the equation~\eqref{eq:LP_tardy_cardconstr} first, and let $g(x_I)= T-\sum_{i,j} x_{ij}$ be the value of the constraint when the decision variables are $x_{ij} = 1$ whenever $(i,j)  \in I$, and 0 otherwise.\footnote{\label{fn:pairs_and_numbers}Strictly speaking $I \subseteq [n]$ is a set of numbers, so we associate to each pair $i,j$ a number via the one-to-one mapping $(i-1)m+j$. Hence, to keep the notation simple, we here understand $(i,j) \in I$ to mean $(i-1)m+j \in I$.} Now for every $I \subseteq [n]$, it holds $g(x_I)y^n_I \geq 0$, as we have $y_I^n = 0$ for every $I$ containing more than $2t+1 = \frac{\sqrt{n}}{k} = T$ elements. Hence the sum in~\eqref{eq:lasserre_alternative_3} is again a sum of positive semidefinite matrices with non-negative weights, and the condition is satisfied.

Finally, consider the $\ell$-th constraint of the form~\eqref{eq:LP_tardy_demand}, and let $g_\ell(x_I) = \sum_{i=1}^l \sum_{j=1}^m x_{ij}\cdot P^i - D_\ell$. In order to prove that~\eqref{eq:lasserre_alternative_3} is satisfied, we apply Lemma~\ref{cor:almost_diagonal_form} with the following collection of subsets of $[n]$: $\mathcal{C} = \set{I \oplus S~|~I \subseteq [m], |I| \leq t}$, where we take $S = \set{(\ell,j) ~|~ j \in [t+1]}$.
%\footnote{See Footnote~\ref{fn:pairs_and_numbers}.}
% In words, $S$ corresponds to the solution where exactly $t+1$ jobs from the block $\ell$ are tardy. Then, an element of $\mathcal{C}$ corresponds to changing at most $t$ jobs in the solution given by $S$ to either tardy or not tardy. Therefore, any solution given by elements of $\mathcal{C}$ contains at least one job from the block $\ell$, meaning that the corresponding allocation $x_I$ satisfies the constraint.
Now, any solution given by the elements of $\mathcal{C}$ contains at least one job from the block $\ell$, meaning that the corresponding allocation $x_I$ satisfies the constraint.

By Lemma~\ref{lem:Z_ts_inverse}, the matrix $Z_{t(S)}$ is invertible and by Lemma~\ref{cor:almost_diagonal_form} we have for \eqref{eq:lasserre_alternative_3} that
$
 \sum_{I \subseteq [n]}g_\ell(x_I)y^n_I Z_IZ_I^\top \cong D + \sum_{I \in [n] \setminus \mathcal{C}} g_\ell(x_I)y^n_I Z_{t(S)}^{-1}Z_I (Z_{t(S)}^{-1}Z_I)^\top,
$
where $D$ is a diagonal matrix with elements $g_\ell(x_I)y^n_I$ for each $I \in \mathcal{C}$. We prove that the latter is positive semidefinite by analysing its smallest eigenvalue $\lambda_{\min}$. Writing $R_I = Z_{t(S)}^{-1}Z_I (Z_{t(S)}^{-1}Z_I)^\top$, we have by Weyl's inequality
$$
\lambda_{\min}\left( D + \sum_{I \in [n] \setminus \mathcal{C}} g_\ell (x_I) y^n_I R_I \right) \geq
\lambda_{\min} \left(D\right) + \lambda_{\min}\left(\sum_{I \in [n] \setminus \mathcal{C}} g_\ell (x_I) y^n_I R_I\right)
$$
Since $D$ is a diagonal matrix with entries $g_\ell (x_I) y^n_I$ for $I \in \mathcal{C}$, and for every $I \in \mathcal{C}$ the constraint $g_\ell(x_I)$ is satisfied, we have $\lambda_{\min}(D)\geq \alpha\left( P^\ell - D_\ell\right) = \alpha \left( P^\ell - \frac{P^\ell - 1}{P-1}\right)$.\

On the other hand for every $I \subseteq [n]$, $g_\ell(x_I) \geq -\sum_{j=1}^\ell P^{j-1} = -\frac{P^\ell-1}{P-1}$. The nonzero eigenvalue of the rank one matrix $R_I$ is $\left(Z_{t(S)}^{-1}Z_I\right)^{\top}Z_{t(S)}^{-1} Z_I\leq |\PS_t(n)|^3 t^{O(t)} = n^{O(\sqrt{n})}$. This is because by Lemma~\ref{lem:Z_ts_inverse_explicite}, for every $I, J \in \PS_t(n)$, $\lvert[Z_{t(S)}^{-1}]_{I,J}\rvert \leq t^{O(t)}$, for $|S|=t+1$, and $ [Z_I]_J \in \{0,1\}$. Thus
$$
\lambda_{\min}\left( D + \sum_{I \in [n] \setminus \mathcal{C}} g_\ell (x_I) y^n_I R_I \right) \geq \alpha \left( P^k - \frac{P^k - 1}{P-1}\right) - \alpha \frac{P^k-1}{P-1} 2^n n^{O(\sqrt{n})} \geq 0
$$
for $P=n^{O(\sqrt{n})}$.

\end{proof}

The above theorem states that the Lasserre hierarchy has an arbitrarily large integrality gap $k$ even at level $t=\frac{\sqrt{n}}{2k}-\frac{1}{2}$. In the following we provide a tight analysis characterization for this instance, namely we prove that the Lasserre hierarchy admits an arbitrarily large gap $k$ even at level $t=\frac{\sqrt{n}}{k}-1$. Note that at the next level, namely $t+1=\sqrt{n}/k$, $\LAS_{t+1}(LP(\sqrt{n}/k))$ has no feasible solution for $k>1$,\footnote{The constraint~\eqref{eq:LP_tardy_demand} implies that any feasible solution for $\LAS_{t+1}(LP(\sqrt{n}/k))$ has $y_I^n=0$ for all $|I| > \sqrt{n}/k$. This in turn implies, with Lemma~\ref{cor:almost_diagonal_form} for $\mathcal{C}=\mathcal{P}_t(n)$, that $ \sum_{I \subseteq [n]} g_\ell(x_I)y^n_I Z_IZ_I^\top \cong D_\ell$, where $D_\ell$ is a diagonal matrix with entries $g_\ell(x_I)y^n_I$, for every $|I| \leq t$, there exists $\ell$ such that $g_\ell(x_I)<0$ which, in any feasible solution implies $y_I^n=0$, contradicting~\eqref{eq:lasserre_alternative_1}.} which gives a tight characterization of the integrality gap threshold phenomenon. The claimed tight bound is obtained  by utilizing a more involved analysis of the sign of the Rayleigh quotient for the almost diagonal matrix characterization of the Lasserre hierarchy.

%\vspace{3cm}
%More precisely, we show the following dichotomy-type result: for any arbitrarily large $k>0$ we provide an instance of the unweighted \textsc{Min Sum of tardy jobs} problem such that the Lasserre hierarchy relaxation has integrality gap $k$ at some level $t$ and no integrality gap at the next level $t+1$. The claimed tight bound is obtained  by utilizing a more involved analysis of the sign of the Rayleigh quotient for the almost diagonal matrix characterization of the Lasserre hierarchy (Theorem~\ref{lem:integrality_gap_tardy_jobs_SILP}).

\begin{theorem}
\label{lem:integrality_gap_tardy_jobs_SILP}
For any $k \geq 1$ and $n$ such that $t=\frac{\sqrt{n}}{k}-1 \in \mathbb{N}$, the following solution is feasible for $\LAS_t(LP(\sqrt{n}/k))$
\begin{eqnarray}\label{eq:sol_for_lem_SILP}
y_I^n=
\left\{ \begin{array}{ll}
\alpha, & \forall I\in \PS_{t+1}(n)\\
0, & \text{otherwise}
\end{array}
\right.
\end{eqnarray}
where $\alpha > 0$ is such that $\sum_{I \subseteq [n]} y_I^n = 1$ and the parameter $P$ is large enough.

\end{theorem}
\begin{proof}

The solution satisfies the conditions~\eqref{eq:lasserre_alternative_1},~\eqref{eq:lasserre_alternative_2} and~\eqref{eq:lasserre_alternative_3} for~\eqref{eq:LP_tardy_cardconstr} by the same argument as in the proof of Theorem~\ref{lem:integrality_gap_tardy_jobs_t-perturbed}.

We prove that the solution satisfies the condition~\eqref{eq:lasserre_alternative_3} for any constraint $\ell$ of the form~\eqref{eq:LP_tardy_demand}. Since $M \succeq 0$ if and only if $v^\top Mv \geq 0 $, for every unit vector $v$ of appropriate size, by Lemma~\ref{cor:almost_diagonal_form} (for the collection $\mathcal{C} = \mathcal{P}_t(n)$) and using the solution~\eqref{eq:sol_for_lem_SILP} we can transform~\eqref{eq:lasserre_alternative_3}
to the following semi-infinte system of linear inequalities
\begin{equation}
\label{eq:tardy_jobs_demand_constrain_SIP}
\sum_{I\in \PS_t(n)} g_\ell(x_I) v^2_I + \sum_{J\subseteq [n]:|J|=t+1} \left(\sum_{\substack{I\in \PS_t(n)\\ I\subset J}} v_I (-1)^{|I|} \right)^2 g_\ell(x_J)  \geq 0, \quad \forall  v\in \mathbb{S}^{|\PS_t(n)|-1}
\end{equation}

Consider the $\ell$-th covering constraint $g_{\ell}(x)\geq 0$ of the form~\eqref{eq:LP_tardy_demand} and the corresponding semi-infinite set of linear inequalities~\eqref{eq:tardy_jobs_demand_constrain_SIP}. Then consider the following partition of $\PS_{t+1}(n)$: $A=\{ I\in\PS_{t+1}(n): I\cap N_{\ell} \not= \emptyset\}$ and $B=\{I\in\PS_{t+1}(n): I\cap N_{\ell}=\emptyset\}$.
%
%\begin{eqnarray*}
%A&=&\{ I\in\PS_{t+1}(n): I\cap N_{\ell} \not= \emptyset\}\\
%B&=&\{I\in\PS_{t+1}(n): I\cap N_{\ell}=\emptyset\}
%\end{eqnarray*}

Note that $A$ corresponds to the assignments that are guaranteed to satisfy the constraint $\ell$. More precisely, for $S\in A$ we have $g_{\ell}(x_S)\geq \left(P^{\ell}-\sum_{j=1}^{\ell}P^{j-1}\right)= P^{\ell}\left(1-\frac{P^{\ell}-1}{P^{\ell}(P-1)}\right)\geq P^{\ell}\left(1-\frac{1}{P-1}\right)$,
and for $S\in B$ we have
$g_{\ell}(x_S)\geq -\sum_{j=1}^{\ell}P^{j-1}\geq P^{\ell}\left(-\frac{1}{P-1}\right)$.
Since $P>0$, by scaling $g_{\ell}(x)\geq 0$ (see~\eqref{eq:LP_tardy_demand}) by $P^{\ell}$, we will assume, w.l.o.g., that
\begin{eqnarray*}
g_{\ell}(x_S)\geq
\left\{
\begin{array}{ll}
1-\frac{1}{P-1}, & \text{if }S\in A\\
-\frac{1}{P-1}, & \text{if } S\in B
\end{array}
\right.
\end{eqnarray*}

Note that, since $v$ is a unit vector, we have $v_I^2\leq 1$, and for any $J\subseteq [n]$ such that $|J|=t+1$, the coefficient of $g_{\ell}(x_J)$ is bounded by $\left(\sum_{\substack{I\in \PS_t(n)\\ I\subset J}} v_I (-1)^{|I|} \right)^2\leq 2^{O(t)}$. For all unit vectors $v$, let $\beta$ denote the smallest possible total sum of the negative terms in~\eqref{eq:tardy_jobs_demand_constrain_SIP} (these are those related to $g_{\ell}(x_I)$ for $I\in B$). Note that $\beta\geq  -\frac{|B|2^{O(t)}}{P}= -\frac{n^{O(t)}}{P}$.

In the following, we show that, for sufficiently large $P$, the claimed solution satisfies \eqref{eq:tardy_jobs_demand_constrain_SIP}. We prove this by contradiction.

Assume that there exists a unit vector $v$ such that~\eqref{eq:tardy_jobs_demand_constrain_SIP} is not satisfied by the solution.
We start by observing that under the previous assumption the following holds $v_I^2 = \frac{n^{O(t)}}{P}$ for all $I\in A\cap \PS_t(n)$.
If not, we would have an $I\in A\cap \PS_t(n)$ such that $v_I^2 g_{\ell}(x_I)\geq -\beta$ contradicting the assumption that \eqref{eq:tardy_jobs_demand_constrain_SIP} is not satisfied.
We claim that under the contradiction assumption, the previous bound on $v_I^2$ can be generalized to $v_I^2=\frac{n^{O(t^2)}}{P}$ for \emph{any} $I\in \PS_t(n)$. Then, by choosing $P$ such that  $v_I^2<1/n^{2t}$, for $I\in \PS_t(n)$, we have $\sum_{I\in \PS_t(n)} v_I^2<1$, which contradicts the assumption that $v$ is a unit vector.

The claim follows by showing that $\forall I\in B\cap  \PS_t(n)$ it holds $v_I^2 \leq n^{O(t^2)}/P$. The proof is by induction on the size of $I$ for any $I \in B\cap  \PS_t(n)$. %We prove that $|v_I| \leq n^2 |v_J|$
%
%More precisely, by induction we prove that for any $I\in B$ with $|I|=i$ it is $v_I^2\leq 2^{O(n^2)}/P(1+N+\ldots+N^i)$.

Consider the empty set, since $\emptyset\in B\cap  \PS_t(n)$. We show that $v_{\emptyset}^2= n^{O(t)}/P$.
With this aim, consider any $J\subseteq N_{\ell}$ with $|J|=t+1$. Note that $J\in A$, so $g_{\ell}(x_J)\geq t+1-1/(P-1)$ and its coefficient $u_J^2=\left(\sum_{\substack{I\in \PS_t(n)\\ I\subset J}} v_I (-1)^{|I|} \right)^2$ is the square of the sum of $v_{\emptyset}$ and other terms $v_I$, all with $I\in A\cap \PS_t(n)$. Ignoring all the other positive terms apart from the one corresponding to $J$ in~\eqref{eq:tardy_jobs_demand_constrain_SIP}, evaluating the sum of all the negative terms as $\beta$ and using a loose bound $g_{\ell}(x_J)\geq 1/2$ for large $P$, we obtain the following bound $b_0$
\begin{equation}\label{eq:indfirst}
|v_{\emptyset}|\leq \sqrt{-2\beta} + \sum_{\emptyset\not =I\subset J} |v_I|\leq b_0= O\left(\sqrt{-\beta}+2^{O(t)} \frac{n^{O(t)}}{\sqrt{P}}\right)= \frac{n^{O(t)}}{\sqrt{P}}
\end{equation}
which implies that $v_{\emptyset}^2=n^{O(t)}/P$.

Similarly as before, consider any singleton set $\{i\}$ with $\{i\}\in B\cap  \PS_t(n)$ and
any $J\subseteq N_{\ell}$ with $|J|=t$. Note that $J\in A$, $g_{\ell}(x_J)\geq t-1/(P-1)$ and its coefficient $u_J^2=\left(\sum_{\substack{I\in \PS_t(n)\\ I\subset J\cup\{i\}}} v_I (-1)^{|I|} \right)^2$ is the square of the sum of $v_{\{i\}}$, $v_{\emptyset}$ and other terms $v_I$, with $I\subseteq J$ and therefore $v_I^2= \frac{n^{O(t)}}{P}$. Moreover, again note that $u_J^2$ is smaller than $-\beta$ (otherwise \eqref{eq:tardy_jobs_demand_constrain_SIP} is satisfied). Therefore, for any singleton set $\{i\}\in B\cap  \PS_t(n)$, we have that
$|v_{\{i\}}|\leq |v_{\emptyset}| + \sqrt{-2\beta} + \sum_{\emptyset\not =I\subset J} |v_I| \leq 2b_0$.

Generalizing by induction, consider any set $S\in B\cap  \PS_t(n)$ and
any $J\subseteq N_{\ell}$ with $|J|=t+1-|S|$. We claim that $|v_{|S|}|\leq b_{|S|}$ where
\begin{equation}\label{eq:recbound}
b_{|S|}= b_0 + \sum_{i=0}^{|S|-1} n^{i} b_{i}
\end{equation}
%The latter \eqref{eq:recbound} 
This follows  by induction hypothesis and by because again $g_{\ell}(J\cup S)u_{J\cup S}\leq -\beta$ and therefore,
%\begin{equation*}
$
|v_{S}|\leq \sum_{i=0}^{|S|-1} \left(\sum_{\substack{I\in B\\ |I|=i}} |v_{I}|\right) + \sqrt{-2\beta} + \sum_{\emptyset \neq I\subset J} |v_I|
$.
%\end{equation*}

From \eqref{eq:recbound}, for any $S\in B\cap  \PS_t(n)$, we have that $|v_{S}|$ is bounded by $b_{t}=(n^{t-1}+1)b_{t-1}=n^{O(t^2)}b_0=\frac{n^{O(t^2)}}{\sqrt{P}}$.

\end{proof}

\section{Application in 0/1 Polynomial Optimization}

%\newpage
%Let $f(x)$ be the objective function of some nonconstrained optimization problem
In this section we use the developed technique to prove an integrality gap result for the unconstrained 0/1 $n$-variate polynomial optimization problem. We start with the following definition of Lasserre hierarchy.

\begin{definition} \label{def:lasserre_uncostrained}
 The Lasserre hierarchy at level $t$ for the unconstrained 0/1 optimization problem with the objective function $f(x):\set{0,1}^n \rightarrow \mathbb{R}$, denoted by $\LAS_t(f(x))$, is given by the feasible points $y^n_I$ for each $I \subseteq [n]$ of the following semidefinite program
\begin{eqnarray}
 \sum_{I \subseteq [n]} y_I^n  &=&  1, \label{eq:lasserre_polynomials_def_1} \\
 \sum_{I \subseteq [n]} y_I^n Z_IZ_I^\top  &\succeq&  0 \text{, where }  Z_I \in \mathbb{R}^{\PS_{t}(n)}  \label{eq:lasserre_polynomials_def_2}
\end{eqnarray}
%where the zeta vectors $Z_I$ are indexed by $\PS_t(n)$.
\end{definition}
The main result of this section is the following theorem.
\begin{theorem}
\label{lem:integrality_gap_polynomials_t-perturbed}
The class of unconstrained $k$--degree 0/1 $n$-variate polynomial optimization problems cannot be solved exactly with a $k-1$ level of Lasserre hierarchy.
%The Lasserre applied for a noncostrained polynomial optimization problem
\end{theorem}
\begin{proof}
For every $k \leq n$ we give an unconstrained $n$-variate polynomial optimization problem with an objective function $f(x)$ of degree $k$ such that $\LAS_{k-1}(f(x))$ has an integrality gap. Consider a maximization problem with the following objective function over $\{0,1\}^n$:
$
f(x) = \sum_{\substack{ I \subseteq [n]\\ |I| \leq k }} \binom{n - |I|}{k - |I|} (-1)^{|I|+1} \prod_{i \in I} x_i.
$
We prove that  the following solution is super-optimal and feasible for $\LAS_{k-1}(f(x))$
\begin{eqnarray}\label{eq:sol_for_polynomials}
y_I^n=
\left\{ \begin{array}{ll}
\alpha, & \forall I\in [n], \phantom{a} |I| \geq n-k+1\\
- \epsilon, & I= \emptyset\\
0, & \text{otherwise}
\end{array}
\right.
\end{eqnarray}
where $\alpha > 0$ is such that $\sum_{\emptyset \neq I \subseteq [n]} y_I^n = 1+\epsilon$ and the $\epsilon$ is small enough.

It is easy to check that the objective function is equivalent to
$$
f(x) = \sum_{\substack{ K \subseteq [n] \\ |K| = k}}\sum_{\substack{ J \subseteq K\\ J \neq \emptyset }} (-1)^{|J|+1}\prod_{j \in J}x_j
$$
Now, consider any integral 0/1 solution, for every $K \subseteq [n]$ of size $|K| = k$, a partial summation $\sum_{\emptyset \neq J \subseteq K} (-1)^{|J|+1}\prod_{j \in J}x_j$ takes value one, if for at least one $j \in K$, $x_j=1$, and zero otherwise. Thus the integral optimum is $\binom{n}{k}$ for any solution $x \in \{0,1\}^n$ such that at least $n-k+1$ coordinates are set to $1$.

On the other hand the objective value for the Lasserre solution~\eqref{eq:sol_for_polynomials} is given by the formula\footnote{The objective function for Lasserre hierarchy at level $t$ is $\sum_{I\in \PS_{2t}(n)} f_I y_I$, for $f$ being the vector of coefficients in $f(x)$. This, with Definition~\ref{def:change_of_basis}, implies the objective function of the form $\sum_{I\subseteq  [n]} f(x_I) y_I^n$.}
 $$
\sum_{I \in [n]}f(x_I) y_I^n=\sum_{\substack{ I \in [n]\\ |I| \geq n-k+1 }}f(x_I) y_I^n = \binom{n}{k}\sum_{\substack{ I \in [n]\\ |I| \geq n-k+1 }} y_I^n =\left(1+\epsilon \right) \binom{n}{k}
$$
where the first equality comes from the fact that $f(x_\emptyset)=0$ and the second from the fact that $f(x_I)=\binom{n}{k}$ for any $I\subseteq [n]$, $|I| \geq n-k+1$.

Finally, we prove that the solution~\eqref{eq:sol_for_polynomials} is feasible for $\LAS_{k-1}(f(x))$. The constraint~\eqref{eq:lasserre_polynomials_def_1} is satisfied by definition. In order to prove that the constraint~\eqref{eq:lasserre_polynomials_def_2} is satisfied, we apply Lemma~\ref{cor:almost_diagonal_form} with the collection
$
\mathcal{C} = \set{I \oplus S~|~I \subseteq [n], |I| \leq k-1}
$
of subsets of $[n]$, for $S = [n]$, and get that
\begin{equation}
\label{eq:perturbed_matrix_form}
D + \sum_{I \in [n] \setminus \mathcal{C}} y^n_I Z_{t(S)}^{-1}Z_I (Z_{t(S)}^{-1}Z_I)^\top =  D  -\epsilon  Z_{t(S)}^{-1}Z_{\emptyset} (Z_{t(S)}^{-1}Z_{\emptyset})^\top
\end{equation}
where $D$ is a diagonal matrix with diagonal entires equal to $\alpha \geq 1/2^n$. Since the nonzero eigenvalue of the rank one matrix $Z_{t(S)}^{-1}Z_{\emptyset} (Z_{t(S)}^{-1}Z_{\emptyset})^\top$ is equal to $\left(Z_{t(S)}^{-1}Z_{\emptyset}\right)^{\top}Z_{t(S)}^{-1} Z_{\emptyset}\leq |\PS_t(n)| t^{2t} = n^{O(t)}$, one can choose $\epsilon = 1/n^{O(t)}$ such that by the Weyl's inequality we have that the matrix in~\eqref{eq:perturbed_matrix_form} is PSD.

\end{proof}
%%%%%%%%%%%%
{\small
\bibliographystyle{abbrv}
\bibliography{tardy_jobs}

\begin{thebibliography}{10}

\bibitem{BansalP10}
N.~Bansal and K.~Pruhs.
\newblock The geometry of scheduling.
\newblock In {\em FOCS}, pages 407--414, 2010.

\bibitem{BarakCK15}
B.~Barak, S.~O. Chan, and P.~Kothari.
\newblock Sum of squares lower bounds from pairwise independence.
\newblock In {\em STOC}, 2015.

\bibitem{BhaskaraCVGZ12}
A.~Bhaskara, M.~Charikar, A.~Vijayaraghavan, V.~Guruswami, and Y.~Zhou.
\newblock Polynomial integrality gaps for strong sdp relaxations of densest
  {\it k}-subgraph.
\newblock In {\em SODA}, pages 388--405, 2012.

\bibitem{CarrFLP00}
R.~D. Carr, L.~Fleischer, V.~J. Leung, and C.~A. Phillips.
\newblock Strengthening integrality gaps for capacitated network design and
  covering problems.
\newblock In {\em SODA}, pages 106--115, 2000.

\bibitem{CheungS11}
M.~Cheung and D.~B. Shmoys.
\newblock A primal-dual approximation algorithm for min-sum single-machine
  scheduling problems.
\newblock In {\em APPROX-RANDOM}, pages 135--146, 2011.

\bibitem{Chla12}
E.~Chlamtac and M.~Tulsiani.
\newblock Convex relaxations and integrality gaps.
\newblock In {\em to appear in Handbook on semidefinite, conic and polynomial
  optimization}. Springer.

\bibitem{FawziSaundersonParrilo15}
H.~Fawzi, J.~Saunderson, and P.~Parrilo.
\newblock Sparse sum-of-squares certificates on finite abelian groups.
\newblock {\em CoRR}, abs/1503.01207, 2015.

\bibitem{Grigoriev01}
D.~Grigoriev.
\newblock Complexity of positivstellensatz proofs for the knapsack.
\newblock {\em Computational Complexity}, 10(2):139--154, 2001.

\bibitem{Grigoriev01b}
D.~Grigoriev.
\newblock Linear lower bound on degrees of positivstellensatz calculus proofs
  for the parity.
\newblock {\em Theoretical Computer Science}, 259(1-2):613--622, 2001.

\bibitem{KarlinMN11}
A.~R. Karlin, C.~Mathieu, and C.~T. Nguyen.
\newblock Integrality gaps of linear and semi-definite programming relaxations
  for knapsack.
\newblock In {\em IPCO}, pages 301--314, 2011.

\bibitem{Karp72}
R.~M. Karp.
\newblock Reducibility among combinatorial problems.
\newblock In {\em Proceedings of a symposium on the Complexity of Computer
  Computations, held March 20-22, 1972, at the {IBM} Thomas J. Watson Research
  Center, Yorktown Heights, New York.}, pages 85--103, 1972.

\bibitem{KurpiszLM15}
A.~Kurpisz, S.~Lepp{\"{a}}nen, and M.~Mastrolilli.
\newblock On the hardest problem formulations for the $0/1$ {L}asserre
  hierarchy.
\newblock In {\em ICALP}, pages 872--885, 2015.

\bibitem{Lasserre01}
J.~B. Lasserre.
\newblock Global optimization with polynomials and the problem of moments.
\newblock {\em SIAM Journal on Optimization}, 11(3):796--817, 2001.

\bibitem{Laurent03}
M.~Laurent.
\newblock A comparison of the {S}herali-{A}dams, {L}ov{\'a}sz-{S}chrijver, and
  {L}asserre relaxations for 0-1 programming.
\newblock {\em Mathematics of Operations Research}, 28(3):470--496, 2003.

\bibitem{MekaPW15}
R.~Meka, A.~Potechin, and A.~Wigderson.
\newblock Sum-of-squares lower bounds for planted clique.
\newblock {\em CoRR}, abs/1503.06447, 2015.
\newblock To appear in STOC 2015.

\bibitem{MestreVerschae14}
J.~Mestre and J.~Verschae.
\newblock A 4-approximation for scheduling on a single machine with general
  cost function.
\newblock {\em CoRR}, abs/1403.0298, 2014.

\bibitem{Moore68}
M.~J. Moore.
\newblock An $n$ job, one machine sequencing algorithm for minimizing the
  number of late jobs.
\newblock {\em Management Science}, 15:102--109, 1968.

\bibitem{ODonnellZ13}
R.~O'Donnell and Y.~Zhou.
\newblock Approximability and proof complexity.
\newblock In S.~Khanna, editor, {\em SODA}, pages 1537--1556. SIAM, 2013.

\bibitem{parrilo00}
P.~Parrilo.
\newblock {\em Structured Semidefinite Programs and Semialgebraic Geometry
  Methods in Robustness and Optimization}.
\newblock {PhD} thesis, California Institute of Technology, 2000.

\bibitem{Rot13}
T.~Rothvo{\ss}.
\newblock The lasserre hierarchy in approximation algorithms.
\newblock Lecture Notes for the MAPSP 2013 - Tutorial, June 2013.

\bibitem{Schoenebeck08}
G.~Schoenebeck.
\newblock Linear level {L}asserre lower bounds for certain k-csps.
\newblock In {\em FOCS}, pages 593--602, 2008.

\bibitem{schor87}
N.~Shor.
\newblock Class of global minimum bounds of polynomial functions.
\newblock {\em Cybernetics}, 23(6):731--734, 1987.

\bibitem{Tulsiani09}
M.~Tulsiani.
\newblock Csp gaps and reductions in the {L}asserre hierarchy.
\newblock In {\em STOC}, pages 303--312, 2009.

\bibitem{Wolsey75}
L.~A. Wolsey.
\newblock Facets for a linear inequality in 0--1 variables.
\newblock {\em Mathematical Programming}, 8:168--175, 1975.

\end{thebibliography}
%,../../bib/MonaldoPublications,../../bib/sdp-ordering}
}

\appendix
\section*{Appendix}

\section{Derivation of the Lasserre hierarchy} \label{app:lasserre}
In this section we derive the formulation of the Lasserre hierarchy used in Section~\ref{sec:Lasserre} and give the missing proofs. In our notation we follow the survey by Rothvo\ss~\cite{Rot13} and we use several known derivations~\cite{Laurent03}.
Let $y \in \mathbb{R}^{\PS_{2t+2}(n)}$ be a vector indexed by the subsets of $\set{1,...,n}$ of size at most $2t+2$, and $M_{t+1}(y)$ the \textit{moment matrix of the variables} $y$ defined by $[M_{t+1}(y)]_{I,J} = y_{I \cup J}$, for $I,J$ subsets of $[n]$ such that $|I|,|J| \leq t+1$. Similarly, for every constraint $\ell$ define the \textit{moment matrix of the constraint} $\ell$ as $[M^\ell_t(y)]_{I,J} = \sum_{i = 1}^n A_{\ell i} y_{I \cup J \cup \set{i}} - b_\ell y_{I \cup J}$, where $|I|,|J| \leq t$.

\begin{definition} \label{def:lasserre}
 The Lasserre hierarchy at level $t$ for the set $K$, denoted by $\LAS_t(K)$, is given by the following semidefinite program
\begin{eqnarray}
 y_\emptyset &=& 1, \\
 M_{t+1}(y) &\succeq& 0, \\
 M_t^\ell(y) &\succeq& 0 \text{ for every constraint }\ell
\end{eqnarray}
\end{definition}

\paragraph{Change of variables.} A point in the Lasserre hierarchy is given by a vector $y \in \mathbb{R}^{\PS_{2t+2}(n)}$, as seen in Definition \ref{def:lasserre}. We now change this variable to a vector that is indexed by \textit{all} the subsets of $[n]$ in order to obtain a useful decomposition of the moment matrix as a sum of rank-one matrices. Here it is not necessary to distinguish between the moment matrix of the variables and constraints, hence in what follows we denote a generic vector by $w \in \mathbb{R}^{\PS_{2d}(n)}$, where $d$ is either $t$ or $t+1$.
\begin{definition}
\label{def:change_of_basis}
 Let $w \in \mathbb{R}^{\PS_{2d}(n)}$. For every $I \in \PS_n(n)$, define a vector $w^n \in \mathbb{R}^{\PS_{n}(n)} $ such that
$$
w_I = \sum_{I \subseteq H \subseteq [n]} w^n_H
$$
\end{definition}

To simplify the notation, we note that the moment matrix of the variables is structurally similar to the moment matrix of the constraints: if $z \in \mathbb{R}^{\PS_{2t}(n)}$ is a vector such that $z_I = \sum_{i = 1}^n A_{\ell i} y_{I\cup \set{i}} - b_\ell y_I$ for some $\ell$, then $[M^\ell_{t}(y)]_{I,J} = z_{I \cup J}$. Hence, the following lemma holds for the moment matrix of variables and constraints.

\begin{lemma} \label{lemma:M_in_zeta_vector_form}
 Let $w \in \mathbb{R}^{\PS_{2d}(n)}$, and $M \in \mathbb{R}^{\PS_{d}(n)\times \PS_{d}(n)}$ such that $M_{I,J} = w_{I \cup J}$. Then
 $$
 M = \sum_{H \subseteq [n]} w^n_H Z_HZ_H^\top
 $$
\end{lemma}
\begin{proof}
Since $M_{I,J} = w_{I \cup J}$, we have by the change of variables that
$$
[M]_{I,J} = \sum_{I \cup J \subseteq H \subseteq [n]} w^n_H = \sum_{H \subseteq [n]} \chi_{I \cup J}(H) w^n_H
$$
where $\chi_{I \cup J}(H)$ is the 0-1 indicator function such that $\chi_I(H) = 1$ if and only if $I \cup J \subseteq H$. On the other hand, $[Z_HZ_H^\top]_{I,J} = [Z_H]_I[Z_H]_J = 1$ if $I \cup J \subseteq H$, and 0 otherwise. Therefore $[Z_HZ_H^\top]_{I,J} = \chi_{I \cup J}(H)$.
\end{proof}

\begin{lemma}
Given $y \in \mathbb{R}^{\PS_{2t+2}(n)}$, for the vector $z_I = \sum_{i = 1}^n A_{\ell i} y_{I\cup \set{i}} - b_\ell y_I$ we have
\begin{equation} \label{eq:z^n_I}
z^n_{I} = g_\ell(x_I) y^n_I
\end{equation}
where $g_\ell(x_I) = \sum_{i = 1}^n A_{\ell i}x_i - b_\ell$ is a linear function corresponding to the constraint $\ell$, evaluated at $x_I$ such that $x_i = 1$ if $i \in I$ and $x_i = 0$ otherwise.
\end{lemma}
\begin{proof}
 We need to show that this choice of $z_I^n$ yields $z_I = \sum_{I \subseteq H \subseteq [n]} z^n_H$. Plug in \eqref{eq:z^n_I}
\begin{eqnarray*}
\sum_{I \subseteq H \subseteq [n]} z^n_H = \sum_{I \subseteq H \subseteq [n]} g_\ell(x_H) y^n_H = \sum_{I \subseteq H \subseteq [n]} \left[\sum_{i = 1}^n A_{\ell i}x_i - b_\ell \right]_{x = x_H} y^n_H \\
= \sum_{I \subseteq H \subseteq [n]}\left( \sum_{i = 1}^n \left[A_{\ell i}x_i\right]_{x = x_H}y^n_H - b_\ell  y^n_H \right) = \sum_{I \subseteq H \subseteq [n]} \sum_{i = 1}^n \left[A_{\ell i}x_i\right]_{x = x_H}y^n_H - b_\ell  y_I
\end{eqnarray*}
Here the term $\left[A_{\ell i}x_i\right]_{x = x_H}y^n_H$ is $A_{\ell i}y^n_H$ if $i \in H$ and 0 otherwise. Taking this into account and changing the order of the sums, the above becomes
$$
\sum_{i = 1}^n \sum_{I \cup \set{i} \subseteq H \subseteq [n]} A_{\ell i}y^n_H - b_\ell  y_I = \sum_{i = 1}^n A_{\ell i} y_{I\cup \set{i}} - b_\ell y_I
$$
which proves the claim.

\end{proof}

The above discussion justifies Definition~\ref{def:lasserre_alternative}.

\end{document}